\documentclass{article}
\usepackage{authblk}
\usepackage{ stmaryrd }
\usepackage{amsmath}
\usepackage{amssymb}
\usepackage{amsfonts}
\usepackage{amsthm}
\usepackage{enumerate}
\usepackage{hyperref}
\usepackage{color}
\usepackage[all]{xy}
\usepackage{tikz}
\usetikzlibrary{matrix,arrows,decorations.pathmorphing}
\theoremstyle{plain}
\newtheorem{thm}{Theorem}[section]

\newtheorem{prop}[thm]{Proposition}

\newtheorem{fact}[thm]{Fact}

\theoremstyle{definition}
\newtheorem{dfn}[thm]{Definition}

\newtheorem{conv}[thm]{Convention}
\newtheorem{exmp}[thm]{Example}

\newtheorem{rem}[thm]{Remark}

\newtheorem{dfns-rems}[thm]{Definitions and Remarks}
\newtheorem{notas-rems}[thm]{Notations and Remarks}
\newtheorem{exmps-rems}[thm]{Examples and Remarks}
\newtheorem*{clm*}{Claim}
\newlength\mylen
\newcommand\RELArrow{
\to\kern-0.55\mylen\vline height 1.1ex depth -0.4pt\kern0.55\mylen}
\newcommand{\rel}{\RELArrow}
\title{Uniform Interpolation for Coalgebraic Fixpoint Logic }

\author{Johannes Marti\thanks{johannes.marti@gmail.com}}
\author{Fatemeh Seifan\thanks{F.Seifan@uva.nl}}
\author{Yde Venema\thanks{Y.Venema@uva.nl}}

\affil{ILLC, Universiteit van Amsterdam, The Netherlands}

\begin{document}
\maketitle
\begin{abstract}
We use the connection between automata and logic to prove that a wide class of coalgebraic fixpoint logics enjoys uniform interpolation. To this aim, first we generalize one of the central results in coalgebraic automata theory, namely closure under projection, which is known to hold for weak-pullback preserving functors, to a more general class of functors, i.e.; functors with quasi-functorial lax extensions. Then we will show that closure under projection implies definability of the bisimulation quantifier in the language of coalgebraic fixpoint logic, and finally we prove the uniform interpolation theorem. 

\end{abstract}

\section{Introduction}
The connection between automata and logic goes back to the early seventies by the works of B{\"u}chi \cite{jb} and Elgot \cite{ce}, who showed that finite automata and monadic second-order logic have the same expressive power over finite words, and that the transformations from formulas to automata and vice versa are effective. This connection has found important applications and landmark results, such as Rabins's decidability theorem \cite{mr}. During the last twenty years study of the link between automata and logic has been continued and many interesting results have been obtained, such as results in \cite{dj:iw}, where Janin and Walukiewicz established the connection between the modal $\mu$-calculus and parity automata operating on labeled transition systems. 

The coalgebraic perspective on the link between automata and logic has been uniformly studied in \cite{yv}, where the author introduces the notion of a coalgebra automaton and establishes the connection between these automata and coalgebraic fixpoint logic based on Moss' modality ($\nabla$)\cite{lm}. Coalgebraic fixpoint logic is a  powerful extension of coalgebraic modal logic \cite{lm} with fixpoint operators. The main contribution of this paper will be to add \emph{uniform interpolation} to the list of properties of coalgebraic fixpoint logic.

A logic has \emph{interpolation} if, whenever we have two formulas $a$ and $b$ such that $\models a\rightarrow b$ (meaning that the formula $a\rightarrow b$ holds in every state of every model), then there is an \emph{interpolant} formula $c$ in the \emph{common language} of $a$ and $b$ (i.e.; $c$ may use only propositional letters that appear both in $a$ and $b$), such that $\models a\rightarrow c$ and $\models c\rightarrow b$. This notion is familiar from first-order logic, and is known there as Craig interpolation \cite{wc} .
Some logics enjoy a much stronger version of interpolation, namely uniform interpolation, which has been introduced by Pitts in \cite{ap}. A logic has uniform interpolation if the interpolant $c$ does not really depend on $b$ itself, but only on the language $b$ shares with $a$. However it is easy to show that classical propositional logic has uniform interpolation, not many logics have this property, for instance first-order logic has interpolation, but it does not have the uniform interpolation \cite{lh}.

In order to provide some more motivation for studying uniform interpolation, let us mention some recent works on this property. Starting with the seminal work of Pitts \cite{ap} who introduced this version of interpolation and proved that intuitionistic logic has uniform interpolation\cite{ap},  the study of this property for different logics has been actively pursued by various authors. In modal logic, Shavrukov \cite{vs} proved that the G{\"o}del-L{\"o}b logic \textbf{GL} has uniform interpolation. Subsequently, Ghilardi \cite{sg} and Visser \cite{av} independently established the property for modal logic \textbf{K}, while \cite{sg:mz} contains negative results for modal logic \textbf{S4}.  In the theory of modal fixtpoint logic D'Agostino and Hollenberg proved that the modal $\mu$-calculus has uniform interpolation \cite{dg:mh}. In the same paper they showed that the logical property of uniform interpolation corresponds to the automata theoretic property of \emph{closure under projection}, which is one of the main results in linking automata and logic.

In this paper we confine our attention to the set functors $\mathsf{T}:\mathsf{Set}\rightarrow\mathsf{Set}$ which have specific kinds of \emph{relation lifting}. A relation lifting $L$ for a functor $\mathsf{T}$ maps every relation $R:X\rel Y$ between the sets $X$ and $Y$ to a relation $LR:\mathsf{T}X\rel\mathsf{T}Y$ between the sets $\mathsf{T}X$ and $\mathsf{T}Y$. The notion of relation lifting has been used in the theory of coalgebras and coalgebraic modal logic to define a notion of bisimilarity (which we denote it by $\leftrightarroweq^{L}$) between states in $\mathsf{T}$-coalgebras, and to define a semantics for the Moss' modality. 

The coalgebraic fixpoint logic $\mu\mathcal{L}_{L}^{\mathsf{T}}(\mathsf{P})$ for functor $\mathsf{T}$ and relation lifting $L$ over set $\mathsf{P}$ of propositional letters, is the extension of coalgebraic modal logic $\mathcal{L}_{L}^{\mathsf{T}}(\mathsf{P})$ with the least fixpoint operator and given by the following grammer:

$$a::= p\mid \neg a\mid \bigvee A\mid\nabla{\alpha}\mid\mu p.a,$$
where $p\in\mathsf{P}$, $A\in\mathcal{P}_{\omega}(\mu\mathcal{L}_{L}^{\mathsf{T}})$ and $\alpha\in \mathsf{T}_\omega (\mu\mathcal{L}_{L}^{\mathsf{T}}(\mathsf{P}))$ (see Definition~\ref{syntax}).
As we have already mentioned, the purpose of this paper is to show that $\mu\mathcal{L}_{L}^{\mathsf{T}}(\mathsf{P})$ enjoys the uniform interpolation. To this aim we will roughly follow the proof by D'Agostino and Hollenberg in \cite{dg:mh} which is based on the definability of a certain nonstandard second-order quantifier in $\mu\mathcal{L}_{L}^{\mathsf{T}}$. This \emph{bisimulation quantifier} is given by the following semantics:
\begin{eqnarray}\label{sem}
\mathbb{S},s\Vdash\exists{p}.b\quad\textmd{iff}\quad\mathbb{S'},s'\Vdash b,~\textmd{for some}~\mathbb{S'},s'~\textmd{with}~\mathbb{S},s\leftrightarroweq_{p}^{L}\mathbb{S'},s',
\end{eqnarray}
where $\leftrightarroweq_{p}^{L}$ denotes the relation of bisimilarity up to the proposition letter $p$. Intuitively (\ref{sem}) says that we can make formula $b$ true by changing the interpretation of $p$, however not necessarily here, but in an up-to-$p$ bisimilar state. For more details on bisimulation quantifiers in modal logic see \cite{tf}.  In the case of coalgebraic modal logic it has been semantically proved in \cite{jm} that if we restrict our attention to the functors $\mathsf{T}$ with a quasi-functorial lax extensions $L$ (special relation liftings, see Def.~\ref{dfn1} for a precise definition), which is a weaker condition than preservation of weak pullbacks, then the bisimulation quantifier is definable and $\mathcal{L}_{L}^{\mathsf{T}}$ (Moss' coalgebraic logic) has uniform interpolation. In addition to  \cite{jm}, in \cite{dp} the author introduced a version of coalgebraic modal logic: the logic of exact covers which enjoys the uniform interpolation. Although the logic of exact covers can be seen as non-monotonic version of Moss' coalgebraic logic that overcomes the requirement of weak pullback preservation, it precludes fixpoint extensions in the style of \cite{yv}. So in this paper we restrict  to the same class of functors as\cite{jm} and we  will extend the result  to the coalgebraic fixpoint logic by proving definability of the bisimulation quantifier via the closure of coalgebra automata under projection. Consequently, the results we will present in this paper can be devided in two types:

\begin{itemize}
\item[(i)] We study purely automata theoretic questions ( Section 4)
\item [(ii)] We apply automata theoretic results to solve questions in logic~(Section 5)
\end{itemize}

Of course, there is an interplay between these two types of results in the following way. We first prove the main technical result of this paper, i.e, the closure of coalgebra automata under projection. This property allows us to show the definabilty of bisimulation quantifiers and finally prove the main result of this paper, the uniform interpolation theorem \ref{uniform}.

\textbf{Overview} We first fix notation and terminology on \textsf{Set}-based functors and coalgebras; we also introduce relation liftings and bisimulations and equip the reader with the necessary background material. In section $3$ we introduce coalgebraic fixpoint logic and give a breif introduction to coalgebra automata theory. After that, we prove in section $4$ our main technical result. We show that if functor $\mathsf{T}:\mathsf{Set}\rightarrow\mathsf{Set}$ has a quasi-functorial lax extension $L$ which preserves diagonals, then $\mathsf{T}$-automata are closed under projection. Finally in section $5$ we combine the results from section $3$ and section $4$ in order to prove uniform interpolation for coalgebraic fixpoint logic $\mu\mathcal{L}_{L}^{\mathsf{T}}$. We finish the paper in section $6$ with an outlook on future results. 

\section{Preliminaries}
This paper presupposes knowledge of the theory of coalgebras \cite{jr}. In this section we recall some of the central definitions to fix the notation. We assume familiarity with basic notions from category theory such as categories, functors, natural transformations and equivalent categories.
\subsection{Set Functors}
We will work in the category \textsf{Set}, that has sets as objects and functions as arrows. It is assumed that the reader is familiar with the usual constructions on sets, so the following explanations are there to fix notation. The notion $f: X\rightarrow Y$ means that $f$ is a function with domain $X$ and codomain $Y$. The identity function for a set $X$ is denoted by $id_X: X\rightarrow X$. The composition of two functions $f: X\rightarrow Y$ and $g: Y\rightarrow Z$ is the usual composition of functions written as $g\circ f: X\rightarrow Z$. For sets $X'\subseteq X$, the inclusion map from $X'$ to $X$ is denoted by $i_{X',X}: X'\hookrightarrow X$, $x\mapsto x$. For a function $f: X\rightarrow Y$ we define the set $Rng(f)=\{y\in Y\mid \exists x\in X, f(x)=y\}\subseteq Y$.
In the following we assume, if not explicitly stated otherwise, that functors are covariant endofunctors in the category \textsf{Set}. 

We first introduce some of the functors that concern us in this paper. The \emph{powerset functor} is the functor $\mathcal{P}:\mathsf{Set}\rightarrow \mathsf{Set}$, which maps a set $S$ to the set of all its subsets $\mathcal{P}(S)=\{V\mid V\subseteq S\}$. A function $f:S\rightarrow T$ is mapped to $\mathcal{P}(f):\mathcal{P}(S)\rightarrow\mathcal{P}(T)$, which is defined for any $V\subseteq S$ by $\mathcal{P}(f)(V)=f[V]=\{f(v)\mid v\in V \}$. The \emph{contravariant powerset} functor $\breve{\mathcal{P}}$ also maps a set $S$ to $\breve{\mathcal{P}}S=\mathcal{P}S$. On functions $\breve{\mathcal{P}}$ is the inverse image map, that is for an $f:S\rightarrow T$ we have $\breve{\mathcal{P}}(f):\breve{\mathcal{P}}(S)\rightarrow\breve{\mathcal{P}}(T)$, $V\mapsto f^{-1}[V]$.
The \emph{neighborhood} functor $\mathcal{N}=\breve{\mathcal{P}}\breve{\mathcal{P}}$ is the double contravariant powerset functor.
Given a set $S$ and an element $\alpha\in \mathcal{N}S$, we define
$$\alpha^{\uparrow}:=\{X\in\mathcal{P}S\mid Y\subseteq ~X ~\text{for}~\text{some}~ Y\in\alpha\},$$
and we say that $\alpha$ is \emph{upward closed} if $\alpha=\alpha^{\uparrow}$.
The \emph{monotone neighborhood} functor $\mathcal{M}$ is the restriction of neighborhood functor to upward closed sets. More concretely the functor $\mathcal{M}$ is given by $\mathcal{M}(S):=\{\beta\in \mathcal{N}(S)\mid \beta~\text{is upward closed}\}$, while for $f:S\rightarrow T$, we define $\mathcal{M}f : \mathcal{M}S\rightarrow \mathcal{M}T$ by $(\mathcal{M}f)(\beta) := ((\mathcal{N}f)(\beta))^{\uparrow}$.
\begin{dfn}
Functor $\mathsf{T}: \textsf{Set}\rightarrow \textsf{Set}$ is \emph{inclusion preserving} if
$\mathsf{T}i_{A,B}=i_{\mathsf{T}A,\mathsf{T}B}$ for all sets $A\subseteq B$.
\end{dfn}
\begin{prop}\label{prop2} Let $\mathsf{T}:Set\rightarrow Set$ be an inclusion preserving functor, then $\mathsf{T}(Rng (f))=Rng\mathsf{T}(f)$ for any function $f$ in $\mathsf{Set}$.
\end{prop}
\begin{proof}
Let $f : A\rightarrow B$ be a set function and $f|_{Rng(f)}:A\rightarrow Rng(f) $ denote the restriction of $f$ on its range. So $f=i\circ f|_{Rng(f)}$, where $i:Rng(f)\hookrightarrow B$ is the inclusion map. Hence $\mathsf{T}(f) = \mathsf{T}(i \circ f|_{Rng (f)}) =\mathsf{T}{i}\circ \mathsf{T}(f|_{Rng (f)} )$ and $Rng(\mathsf{T}f) = Rng (\mathsf{T}{i}\circ \mathsf{T}(f|_{Rng(f)}))$. But the inclusion preserving property of $\mathsf{T}$ implies that $\mathsf{T}i$ is the inclusion map from $\mathsf{T}(Rng(f))$ to $\mathsf{T}B$, and from this fact it follows that $Rng (\mathsf{T}{i}\circ \mathsf{T}(f|_{Rng (f)}))=Rng(\mathsf{T}(f|_{Rng(f)}))$. But since $f|_{Rng(f)}$ is a surjective map and set functors preserve the surjectiveness of functions, we have $Rng(\mathsf{T}(f|_{Rng(f)}))=\mathsf{T}(Rng (f))$. Hence $\mathsf{T}(Rng (f))=Rng\mathsf{T}(f)$.
\end{proof}
In the context of coalgebraic logic one pays special attention to functors that preserve finite sets and are finitary.
\begin{dfn}
A functor $\mathsf{T}$ \emph{preserves finite sets} if $\mathsf{T}X$ is finite whenever $X$ is. An inclusion preserving $\mathsf{T}$ is called \emph{finitary} if it satisfies for all sets $X$
$$\mathsf{T} X = \bigcup\{\mathsf{T}X'\subseteq\mathsf{T}X\mid X'\subseteq X, ~X' ~\text{is finite}\}.$$
\end{dfn}
The definition can be simply generalized to the class of all set functors as follows:
A set functor $\mathsf{T}$ is finitary if it satisfies for all sets $X$
$$\mathsf{T} X = \bigcup\{\mathsf{T}i_{X', X}(\mathsf{T}X')\subseteq\mathsf{T}X\mid X'\subseteq X, ~X' ~\text{is finite}\}.$$
For every inclusion preserving set functor $\mathsf{T}$ one can define its finitary version $\mathsf{T}_{\omega}$ such that it maps a set $X$ to
$$\mathsf{T}_{\omega}X=\bigcup\{\mathsf{T}X'\mid X'\subseteq X, X' \text{is finite}\}.$$
Similarly, the finitary version of an arbitrary set functor $\mathsf{T}$ can be defined such that it maps a set $X$ to
$$\mathsf{T}_\omega X = \bigcup\{\mathsf{T}i_{X', X}(\mathsf{T}X')\subseteq\mathsf{T}X\mid X'\subseteq X, ~X' ~\text{is finite}\}.$$\\
A function $f:X\rightarrow Y$ is mapped by $\mathsf{T}_\omega$ to the function
\begin{eqnarray*}
\mathsf{T}_{\omega}f:\mathsf{T}_{\omega}X &\rightarrow &\mathsf{T}_{\omega}X, \\
\xi &\mapsto & \mathsf{T}i_{f[X'],Y}\circ \mathsf{T}f_{X'}(\xi')
\end{eqnarray*}
where $\xi '\in \mathsf{T}X'$ is such that $\xi=i_{X', X}(\xi ')$ for a finite $X'\subseteq X$ and $f_{X'}$ is the function $f_{X'}:X'\rightarrow f[X']$, $x'\mapsto f(x).$
An example of a finitary version of a functor that we will use is $\mathcal{P}_\omega$, that maps a set $X$ to the set of all its finite subsets.
An other important class of set functors in the context of coalgebraic modal logic is the class of intersection preserving functors.
\begin{dfn} A set functor $\mathsf{T}$ \emph{preserves finite intersections} if for all sets $A$ and $B$, $\mathsf{T}(A\cap B)=\mathsf{T}A\cap\mathsf{T}B.$
\end{dfn}
In ~\cite[Proposition 2.1]{vt} it has been shown that every set functor preserves non-empty finite intersections, which means if $A\cap B\not=\emptyset$ then we have $\mathsf{T}(A\cap B)=\mathsf{T}A\cap\mathsf{T}B$ for free. But the proof doesn't work for empty intersections and there are some functors, like the monotone neighborhood functor, which do not preserve the empty intersection. However it may not be the case that a given set functor $\mathsf{T}$ preserves all (empty and non-empty) finite intersections, one can redefine $\mathsf{T}$ on the empty set and on the empty maps to obtain a functor $\mathsf{T}'$ which preserves all finite intersections:
\begin{prop}\label{redef} Every set functor $\mathsf{T}$ preserves non-empty finite intersections. By redefining $\mathsf{T}$ on the empty set $\emptyset$ and on the empty maps $\emptyset_A:\emptyset\rightarrow A$, it can be made to preserve all (empty and non-empty) finite intersections.
\end{prop}
\begin{proof}The elementary proofs for the fact that $\mathsf{T}$ preserves non-empty intersections, as we mentioned above, can be found in \cite{vt} or \cite{ja:vt}. In order to modify $\mathsf{T}$ on the empty set and on the empty mappings, Trnkov\'{a} considers first the functor $C_{0,1}$, which maps the empty set to itself and every non-empty set to the one-element set $\{*\}$.
Let $\mathsf{T}'$ agree with $\mathsf{T}$ everywhere, except on the empty set and on the empty mappings. $\mathsf{T}'(\emptyset)$ is defined to be the set of all natural transformations $\nu:C_{0,1}\rightarrow \mathsf{T}$. For each empty map $\emptyset_X:\emptyset\rightarrow X$, whenever $X\not=\emptyset$, define $\mathsf{T}'\emptyset_{X}$ by $\mathsf{T}'\emptyset_{X}(\nu):=\nu_X(\{*\})$ for each $\nu$. Then $\mathsf{T}'$ preserves all finite intersections and $\mathsf{T}'\emptyset_{X}$ is injective for each set $X$. This, together with the fact that all set functors preserves injections with non-empty domain~\cite{mb}, implies that $\mathsf{T}'$ preserves all injections.
\end{proof}
\subsection{Coalgebras}
In the following part of this section, we will briefly recall the basic notions from the theory of coalgebras that we will use later. For a detailed introduction into coalgebras see for example \cite{jr}.
\begin{dfn}Given a set functor $\mathsf{T}$, a $\mathsf{T}$-\emph{coalgebra} is a pair $\mathbb{S}=(S,\sigma)$ with $\sigma: S\rightarrow\mathsf{T}S$.
A pointed $\mathsf{T}$-coalgebra is a pair consisting of a $\mathsf{T}$-coalgebra together with an element of (the carrier set of) that coalgebra. A $\mathsf{T}$-\emph{coalgebra} \emph{morphism} from $\mathsf{T}$-coalgebra $\mathbb{S}=(S,\sigma)$ to $\mathbb{S'}=(S',\sigma')$, written $f:\mathbb{S}\rightarrow\mathbb{S'}$, is a function $f:S\rightarrow S'$ such that $\mathsf{T}(f)\circ\sigma=\sigma'\circ f$.
\end{dfn}
It is easy to check that the collection of $\mathsf{T}$-coalgebra morphisms contains all identity arrows and is closed under arrow composition. So, the $\mathsf{T}$-coalgebras with their morphisms form a category denoted by $Coalg(\mathsf{T})$.
\begin{dfn}
Let $\mathsf{T}$ be an endofunctor on the category $\textsf{Set}$, and $C$ an arbitrary set of objects that we shall call \emph{colors}. We let $\mathsf{T}_C$ denote the functor $\mathsf{T}_C S=\mathsf{T}S\times C$; that is, $\mathsf{T}_C$ maps a set $S$ to the set $\mathsf{T}S\times C$ (and a function $f:S\rightarrow S' $ to the function $\mathsf{T}f\times id_C:\mathsf{T}S\times C\rightarrow \mathsf{T}S'\times C)$. $\mathsf{T}_C$-coalgebras will also be called $C$-\emph{colored} $\mathsf{T}$-\emph{coalgebras}. We will usually denote $\mathsf{T}_C$-coalgebras as triples $\mathbb{S}=(S, \sigma, \gamma)$, with $\sigma: S\rightarrow\mathsf{T}S$ the coalgebra map and $\gamma:S\rightarrow C$ the \emph{coloring (marking)}.
\begin{conv}\label{int} Since the modification of a functor $\mathsf{T}$ to functor $\mathsf{T'}$, given by Trnkov\'{a} in the proof of Proposition~\ref{redef} is not going to change the $\mathsf{T}$-coalgebras i.e., the category $Coalg(\mathsf{T})$ is equivalent to the category $Coalg(\mathsf{T'})$, we will from now on assume that the functor $\mathsf{T}$ preserves all finite intersections and inclusions.
\end{conv}
\end{dfn}
\subsection{Relation Lifting and Bisimulation}
In the remaining part of this section we introduce the notion of relation lifting to define a very general notion of bisimulation for coalgebras. First we recall some central definitions and fix mathematical notation and terminology.
Given sets $X$ and $Y$, we denote a relation $R$ between $X$ and $Y$ by $R:X\rel Y$ to specify its domain $X$ and codomain $Y$. We write $R;S:X\rel Z$ for the composition of two relations $R:X\rel Y$ and $S:Y\rel Z$ and $R^{\circ}:Y\rel X$ for the converse of $R:X\rel Y$ with $(y,x)\in R^{\circ}$ iff $(x,y)\in R$. The graph of any function $f:X\rightarrow Y$ is a relation $f:X\rel Y$ between $X$ and $Y$ for which we also use the symbol $f$. It will be clear from the contex in which a symbol $f$ occurs whether it is meant as a function or a relation. Note that the composition of functions is denoted the other way round the composition of relations, so we have $g\circ f=f;g$ for functions $f:X\rightarrow Y$ and $g:Y\rightarrow Z$.
For a relation $R:X\rel Y$ we define the sets $$Dom(R)=\{x\in X \mid \exists y\in Y , (x,y)\in R\}\subseteq X,$$ $$Rng(R)=\{y\in Y\mid \exists x\in X, (x,y)\in R\}\subseteq Y.$$
The relation $R:X\rel Y$ is \emph{full} on $X$ if $Dom(R)=X$ and is full on $Y$ if $Rng(R)=Y$. Given sets $X'\subseteq X$ and $Y'\subseteq Y$, we define the restriction $R\mid_{X'\times Y' }:X'\rel Y'$ of the relation $R:X\rel Y$ as $R\mid_{X'\times Y' } =R\cap (X′\times Y′)$.
For any set $X$ let $\in_{X} : X\rightarrow \mathcal{P}X$ be the \emph{membership relation} between elements of $X$ and subsets of $X$.
For a relation $R:X\rel Y$ we define the relation $\hat{R}:X\times Y\rel Y$ by $(x,y)\in R$ iff $((x,y),y)\in\hat{R}$. From the definition it is obvious that $\hat{R}$ is functional.
Given a set $X$ we define the \emph{diagonal relation} $\Delta_{X}:X\rel X$ with $(x,x')\in \Delta_{X}$ iff $x=x'$. Note that $\Delta_{X}=id_{X}$, where $id_{X}$ is the graph of the identity function $id_{X}$.
\begin{dfn}
A relation lifting $L$ for a set functor $\mathsf{T}$ is a collection of relations $LR$ for every relation $R$, such that $ LR : \mathsf{T}X\rel\mathsf{T}Y$ if
$R:X\rel Y$ . We require relation liftings to preserve converse, this means that $L(R^\circ) = (LR)^{\circ}$ for all relations $R$.
\end{dfn}

\begin{exmp} (i) The \emph{Egli-Milner lifting} $\overline{\mathcal{P}}$ is a relation lifting for covariant power set functor $\mathcal{P}$ that is defined for any $R:X\rel Y$ such that $\overline{\mathcal{P}}R=\overrightarrow{\mathcal{P}}R\cap \overleftarrow{\mathcal{P}}R$, where:
$$\overrightarrow{\mathcal{P}}R:=\{(U,V)\in\mathcal{P}X\times\mathcal{P}Y\mid\forall u\in U ~ \exists v\in V~\text{s.t.}~ (u,v)\in R\},$$
$$\overleftarrow{\mathcal{P}}R:=\{(U,V)\in\mathcal{P}X\times\mathcal{P}Y\mid\forall v\in V ~ \exists u\in U~\text{s.t.}~ (u,v)\in R\}.$$

(ii) For the constant functor $D$ of a fixed set $D$ define a relation lifting $\overline{D}$ for any $R:X\rel Y$ such that $\overline{D}R=\Delta_D$.

(iii) Recall the notion of $\overrightarrow{\mathcal{P}}R$ from (i) we can define a relation lifting $\widetilde{\mathcal{M}}$ for the monotone neighborhood functor $\mathcal{M}$ on a relation $R:X\rel Y$ as follows:
$$\widetilde{\mathcal{M}}R:=\overrightarrow{\mathcal{P}}\overleftarrow{\mathcal{P}}R\cap \overleftarrow{\mathcal{P}}\overrightarrow{\mathcal{P}}R.$$
\end{exmp}

An important use of relation liftings is to yield a notion of bisimulation.
\begin{dfn}
Let $L$ be a relation lifting for the functor $\mathsf{T}$ and $\mathbb{S}=(S,\sigma)$ and $\mathbb{S'}=(S',\sigma')$ be two $\mathsf{T}$-coalgebras. An $L$-\emph{bisimulation} between $\mathbb{S}$ and $\mathbb{S'}$ is a relation $R:S\rel{S'}$ such that $(\sigma(s),\sigma'(s'))\in LR$, for all $(s,s')\in R$. A state $s\in\mathbb{S}$ is $L$-bisimilar to a state $s'\in\mathbb{S'}$ if there is an $L$-bisimulation $R$ between $\mathbb{S}$ and $\mathbb{S'}$ with $(s,s')\in R$. We write $\leftrightarroweq^{L}$ for the notion of $L$-bisimulation between two fixed coalgebras.
Given two $C$-colored $\mathsf{T}$-coalgebras $\mathbb{S}=(S, \sigma, \gamma)$ and $\mathbb{S'}=(S', \sigma', \gamma')$ and a relation lifting $L$ for the functor $\mathsf{T}$, a relation $R:S\rel{S'}$ is called an $L_{C}$-\emph{bisimulation} between $\mathbb{S}$ and $\mathbb{S'}$, whenever $(\sigma(s),\sigma'(s'))\in LR$ and $\gamma(s)=\gamma '(s')$ for all $(s,s')\in R$.
\end{dfn}

Now we will give the definition of lax extensions, which are relation liftings satisfying certain conditions that make them well-behaved in the context of coalgebra.
\begin{dfn}\label{dfn1}
A relation lifting $L$ for a functor $\mathsf{T}$ is called a \emph{lax extension} of $\mathsf{T}$ if it satisfies the following conditions, for all relations $R, R':X\rel Z$ and $S:Z\rel Y$ and all functions $f:X\rightarrow Z$
\begin{itemize}
\item[(L1)] $R'\subseteq R$ implies $LR'\subseteq LR,$
\item[(L2)] $LR;LS\subseteq L(R;S),$
\item[(L3)] $\mathsf{T}f\subseteq Lf.$
\end{itemize}
We say that a lax extension $L$ \emph{preserves diagonals} if it additionally satisfies:
\begin{itemize}
\item[(L4)]$L\Delta_X\subseteq {\Delta}_{\mathsf{T}X}.$
\end{itemize}
\end{dfn}
We call a lax extension $L$ of $\mathsf{T}$ \emph{functorial}, if it distributes over composition, i.e., for all relations $R:X \rel Z$ and $S:Z \rel Y$, $L(R);L(S)= L(R;S)$. A lax extension $L$ is called \emph{quasi-functorial}, if it satisfies the following condition for all relations $R:X \rel Z$ and $S:Z \rel Y$:
$$L(R);L(S)= L(R;S)\cap(Dom(LR)\times Rng(LS)).$$

\begin{exmp} The relation lifting $\widetilde{\mathcal{M}}$ for the monotone neighborhood functor is quasi-functorial. It is easy to check that $\widetilde{\mathcal{M}}$ is a lax extension that preserves diagonals. So in the following we will just give the proof for the quasi-functoriality of $\widetilde{\mathcal{M}}$:

Take any two relations $R:X\rel Z$ and $S:Z\rel Y$. We need to show that for all $(\alpha, \beta)\in \widetilde{\mathcal{M}}$, if there are $\gamma_R$ and $\gamma_S$ in $\mathcal{M}Z$, with $(\alpha,\gamma_R)\in \widetilde{\mathcal{M}}R$ and $(\gamma_S,\beta)\in \widetilde{\mathcal{M}}S$, then there is a $\gamma\in\mathcal{M}Z$ such that $(\alpha, \gamma)\in \widetilde{\mathcal{M}}R$ and $(\gamma, \beta)\in \widetilde{\mathcal{M}}S$.

From the assumption that  $(\alpha,\gamma_R)\in \widetilde{\mathcal{M}}R\subseteq \overrightarrow{\mathcal{P}}\overleftarrow{\mathcal{P}}R$, we get that: 
$$\forall A\in\alpha, \exists U_a\in\gamma_R~\text{s.t.}~ (A,U_a)\in\overleftarrow{\mathcal{P}}R.$$
Similarly we get the followings:

\vspace{4 mm}

$(\alpha,\beta)\in\widetilde{\mathcal{M}}(R;S)\subseteq\overrightarrow{\mathcal{P}}\overleftarrow{\mathcal{P}}(R;S)$ imply that:
$$\forall A\in\alpha, \exists V_A\in\beta~\text{s.t.}~ (A,V_A)\in\overleftarrow{\mathcal{P}}(R;S).$$

$(\gamma_S,\beta)\in\widetilde{\mathcal{M}}S\subseteq\overleftarrow{\mathcal{P}}\overrightarrow{\mathcal{P}}S$ and $(\alpha,\beta)\in\widetilde{\mathcal{M}}(R;S)\subseteq\overleftarrow{\mathcal{P}}\overrightarrow{\mathcal{P}}(R;S)$ implies that:
$$\forall B\in\beta, \exists U_B\in\gamma_S~\text{and}~\exists A_B\in\alpha ~\text{s.t.}~ (U_B,B)\in\overrightarrow{\mathcal{P}}S~\text{and}~ (A_B,B)\in\overrightarrow{\mathcal{P}}(R;S) .$$
From $(A,V_A)\in\overrightarrow{\mathcal{P}}(R;S)$ it follows that:
$$\forall v\in V_A, \exists a_v\in A ~\text{s.t.}~ (a_v,v)\in (R;S),$$
so there is a $z_v\in Z$ such that $(a_v,z_v)\in R$ and $(z_v,v)\in S$.

Now we define for every $A\in\alpha$:
$$U'_A=U_A\cup\{z_v\in Z\mid v\in V_A\}.$$
We claim that $(A,U'_A)\in \overleftarrow{\mathcal{P}}R$. To prove this, take $u\in U'_A$, we need to show that there exists $t\in A$ such that $(t,u)\in R$. Since $u\in U'_A$, we have two cases:
\begin{itemize}
\item[(i)] $u\in U_A$, then from $(A, U_A)\in \overleftarrow{\mathcal{P}}R$ we are done.
\item[(ii)] $u\in\{z_v\in Z\mid v\in V_A\}$. In this case from the definition of $z_v$ we have that there exists $a_v\in A$ such that $(a_v,z_v)\in R$.
\end{itemize}
On the other hand because $\forall v\in V_A, (z_v, v)\in S$, we have that $(U'_A, V_A)\in \overleftarrow{\mathcal{P}}S$.
 
We can similarly define for every $B\in\beta$ a set $U'_B$ such that:
$$(U'_B,B)\in \overrightarrow{\mathcal{P}}S~\text{and}~ (A_B, U'_B)\in\overrightarrow{\mathcal{P}}R.$$

Now we are ready to introduce $\gamma\in \mathcal{M}Z$:

$$\gamma=\{U\subseteq Z\mid \exists A\in \alpha ~\text{with}~ U'_A\subseteq U ~\text{or}~\exists B\in \beta ~\text{with}~ U'_B\subseteq U\}$$

It is clear that $\gamma$ is upward closed, so $\gamma\in \mathcal{M}Z$.
It is left to show that $(\alpha,\gamma)\in \widetilde{\mathcal{M}}R$ and $(\gamma,\beta)\in \widetilde{\mathcal{M}}S$. We have that
$(\alpha,\gamma)\in \widetilde{\mathcal{M}}R$ iff $(\alpha,\gamma)\in\overrightarrow{\mathcal{P}}\overleftarrow{\mathcal{P}}R $ and $(\alpha,\gamma)\in\overleftarrow{\mathcal{P}}\overrightarrow{\mathcal{P}}R$. 
For the proof of $(\alpha,\gamma)\in\overrightarrow{\mathcal{P}}\overleftarrow{\mathcal{P}}R $ note that for every $A\in \alpha$ we have that $(A, U'_A)\in\overleftarrow{\mathcal{P}}R$.
For $(\alpha,\gamma)\in\overleftarrow{\mathcal{P}}\overrightarrow{\mathcal{P}}R$, pick $U\in\gamma$. Then from the definition of $\gamma$ it follows that $U'_A\subseteq U$ for  $A\in\alpha$ or $U'_B\subseteq U$ for $B\in\beta$.
In the first case consider that $U_A\subseteq U'_A\subseteq U$ and by the assumption $(\alpha,\gamma_R)\in\widetilde{\mathcal{M}}R\subseteq\overleftarrow{\mathcal{P}}\overrightarrow{\mathcal{P}}R$ we get that there exists $T\in \alpha$ such that $(T, U_A)\in \overrightarrow{\mathcal{P}}R$, so $(T, U)\in\overrightarrow{\mathcal{P}}R$.
For the case that there exists $B\in \beta$ such that $U'_B\subseteq U$, we have that $(A_B; U'_B)\in\overrightarrow{\mathcal{P}}R $. 
\end{exmp}

\begin{prop}\label{diag}
Let $\mathsf{T}$ be a set functor and let $L$ be a quasi-functorial lax extension for $\mathsf{T}$. Then we have:
\begin{itemize}
\item[(1)] L preserves fullness of relations :\\
If $R: X\rel Z$ is full on both sides, then so is $LR:\mathsf{T}X \rel\mathsf{T}Z$;
\item[(2)] If $R:X\rel Z$ is full on X and $i:Z\hookrightarrow Z'$ is the inclusion map between $Z$ and $Z'$ then $L(R;i)$ is full on $\mathsf{T}X$;
\item[(3)] If $L$ preserves diagonals then for any function $f$, $\mathsf{T}f= Lf$.
\end{itemize}
\end{prop}
\begin{proof} 
For the proof of (1) consider the following argument:
Let $\pi_X: R\rightarrow X$ and $\pi_Z: R\rightarrow Z$ denote the projection maps. Since $R=(\pi_X)^{\circ};\pi_Z$ and $LR=L((\pi_X)^{\circ};\pi_Z)$, from quasi-functoriality of $L$ it follows that
$$L(\pi_X)^{\circ};L\pi_Z=L((\pi_X)^{\circ};\pi_Z)\cap (Dom L(\pi_X)^{\circ}\times Rng(L\pi_Z)).$$
But since $R=(\pi_X)^{\circ};\pi_Z$ is full on both sides, the projection maps $\pi_X$ and $\pi_Z$ are surjective. It then follows that $\textsf{T}\pi_X:\textsf{T}R\rightarrow \textsf{T}X$ and $\textsf{T}\pi_Z:\textsf{T}R\rightarrow \textsf{T}Z$ are surjective, because set functors preserve surjective-ness. So $Rng(\textsf{T}\pi_X)=Dom(\textsf{T}\pi_X)^{\circ}=\textsf{T}X$ and $Rng(\textsf{T}\pi_Z)=\textsf{T}Z$. Consequently we have
$$L((\pi_X)^{\circ};\pi_Z)\cap \textsf{T}X\times \textsf{T}Z=L(\pi_X)^{\circ};L\pi_Z,$$
which simply implies $L((\pi_X)^{\circ};\pi_Z)=L(\pi_X)^{\circ};L\pi_Z$. Now in order to prove fullness of $LR=L((\pi_X)^{\circ};\pi_Z)$ on $\textsf{T}X$ and $\textsf{T}Z$ it is sufficient to prove $L(\pi_X)^{\circ} ;L\pi_Z :\textsf{T}X\rel \textsf{T}Z$ is full on $\textsf{T}X$ and $\textsf{T}Z$. But we are done since
\begin{eqnarray}
\textsf{T}X &=& Dom((\textsf{T}\pi_X)^{\circ};(\textsf{T}\pi_Z)) \quad\quad \text{$\textsf{T}\pi_X$ is surjective} \nonumber \\
&\subseteq& Dom(L(\pi_X)^{\circ};(L\pi_Z)) \quad\quad \text{($L3$) } \nonumber
\end{eqnarray}
and
\begin{eqnarray}
\textsf{T}Z &=& Rng((\textsf{T}\pi_X)^{\circ};(\textsf{T}\pi_Z)) \quad\quad \text{$\textsf{T}\pi_Z$ is surjective} \nonumber \\
&\subseteq& Rng(L(\pi_X)^{\circ};(L\pi_Z)) \quad\quad \text{($L3$) } \nonumber
\end{eqnarray}

To prove $(2)$ notice that $R;i\subseteq X\times Y$ is full on $X$, so by axiom of choice there exists a map $f:X\rightarrow Y$ such that $f\subseteq (R;i)$. Hence we get $\mathsf{T}f\subseteq Lf\subseteq L(R;i)$, and because $\mathsf{T}f$ is full on $\mathsf{T}X$, $L(R;i)$ is also full on $\mathsf{T}X$.

\begin{equation*}
\begin{xy}
(15,15)*+{(R;i)}="a";%
(0,0)*+{Y}="c"; (30,0)*+{X}="b";%
{\ar@{->>}@/^0.5pc/^{\pi_{X}} "a";"b"};%
{\ar@{->}@/_0.5pc/_{\pi_{Y}} "a";"c"};%
{\ar "b";"c"}?*!/_3mm/{f};
\end{xy}
\end{equation*}

For the proof of $(3)$ we refer to ~\cite[Proposition 2]{jm:yv} (where if fact it is stated that $(3)$ holds for every lax extension $L$).
\end{proof}

Let us now summarize two facts that we will need about $L$-bisimulations in the sequel.
\begin{prop}
For a lax extension $L$ of $\mathsf{T}$ and $\mathsf{T}$-coalgebras $\mathbb{S}$, $\mathbb{S'}$ and $\mathbb{Q}$ the following hold:
\begin{itemize}
\item[(1)] The graph of a coalgebra morphism $f$ from $\mathbb{S}$ to $\mathbb{S'}$ is an $L$-bisimulation between $\mathbb{S}$ and $\mathbb{S'}$;
\item[(2)] if $R:S\rel Q$ respectively $R' :Q\rel S'$ are $L$-bisimulations between $\mathbb{S}$ and $\mathbb{Q}$ respectively $\mathbb{Q}$ and $\mathbb{S'}$, then $R;R':S\rel S'$ is an $L$-bisimulation between $\mathbb{S}$ and $\mathbb{S'}$.
\end{itemize}
\end{prop}
For the proof we refer to ~\cite[Proposition 3]{jm:yv}.

We will finish this section with a remark on some of the closure properties of the class of functors with a quasi-functorial lax extension:

\begin{fact} The collection of functors with a quasi-functorial lax extension \text{(FQL)} has the following properties:
\begin{itemize}
\item[(i)] the identity functor $I:\textsf{Set}\rightarrow \textsf{Set}$ is in \text{FQL};
\item[(ii)]for each set $D$, the constant functor $D:\textsf{Set}\rightarrow \textsf{Set}$ is in \text{FQL};
\item[(iii)] the product $X\mapsto\mathsf{T}_{1}(X)\times\mathsf{T}_{2}(X)$ of tow \text{FQL}s $\mathsf{T}_{1}$ and $\mathsf{T}_{2}$ is in \text{FQL};
\item[(iv)]the coproduct $X\mapsto\mathsf{T}_{1}(X)+\mathsf{T}_{2}(X)$ of tow \text{FQL}s $\mathsf{T}_{1}$ and $\mathsf{T}_{2}$ is in \text{FQL};

\item[(v)]the composition $X\mapsto(\mathsf{T}_{1}\circ\mathsf{T}_{2})(X)$ of a \text{FQL} functor $\mathsf{T}_{1}$ and a functor $\mathsf{T}_{2}$ which has a functorial lax extension, is  in \text{FQL}.
\end{itemize}
\end{fact}
\begin{proof}
 Here we will give a proof for item $(v)$.
Suppose that $L_1$ is a quasi-functorial lax extension for $\mathsf{T}_1$ and $L_2$ is a functorial lax extension for $\mathsf{T}_2$. We claim that $L_{1}L_2$ is a quasi-functorial lax extension for $\mathsf{T}_1\circ\mathsf{T}_2$. First observe that since $L_1$ and $L_2$ are lax extensions, $L_{1}L_2$ is also a lax extension. Take $(\alpha, \beta)\in L_{1}L_2(R;S)\cap Dom(L_{1}L_2R)\times Rng(L_{1}L_2S).$ We get that  $(\alpha, \beta)\in Dom(L_{1}(L_2R))\times Rng(L_{1}(L_2S))$  and by  functoriality of $L_2$, $(\alpha, \beta)\in L_{1}(L_{2}R ; L_{2}S)$. Now from quasi-functoriality of $L_1$ we get that:
 
$$(\alpha , \beta)\in L_1(L_2R);L_1(L_2S)=L_1L_2R;L_1L_2S.$$ 
\end{proof}

\section{Coalgebraic Fixpoint Logic and Automata}
\subsection{Coalgebraic Fixpoint Logic}
In this section we show how to define the syntax and semantics of a coalgebraic fixpoint logic, using a quasi-functorial lax extension $L$ of $\mathsf{T}$. For this purpose from now on we fix a functor $\mathsf{T}$ with a quasi-functorial lax extension $L$. Recall that by our convetion~\ref{int} $\mathsf{T}$ preserves all inclusions and finite intersections. We also fix a set $\mathsf{P}$ of propositional letters and assume that $L$ preserves diagonals.
Before going to the definition of coalgebraic fixpoint logic, we need an auxiliary definition.
\begin{dfn}\label{base}
Given a functor $\textsf{T}$, we define for every set $X$ the function
\begin{eqnarray}
Base: \textsf{T}_{\omega}X&\rightarrow& \mathcal{P}_{\omega}X \nonumber \\
\alpha &\mapsto& \bigcap\{X'\subseteq X\mid \alpha\in \textsf{T}X'\}. \nonumber
\end{eqnarray}
\end{dfn}
This is well-defined because for any $\alpha\in\textsf{T}_{\omega}X$ there is a finite $X''\subseteq X$ such that $\alpha\in\textsf{T}X''$. The definition is useful because for all $\alpha\in\textsf{T}_{\omega}X$ we have that $Base(\alpha)\in \mathcal{P}_{\omega}X $ is the least set $U\in\mathcal{P}_{\omega}X $ such that $\alpha\in\textsf{T}U$. The existence of such a $U$ is given by the fact that $\textsf{T}$ preserves finite intersections[Convention~\ref{int}].
The language of the coalgebraic fixpoint logic $\mu\mathcal{L}^{\mathsf{T}}_L(\mathsf{P})$ is defined as follows:
\begin{dfn} \label{syntax}
For $\mathsf{P}$ as the set of propositional letters, define the language $\mu\mathcal{L}_{L}^{\mathsf{T}}(\mathsf{P})$ by the following grammar:
$$a::= p\mid \neg a\mid \bigvee A\mid\nabla{\alpha}\mid\mu p.a,$$
where $p\in\mathsf{P}$, $A\in\mathcal{P}_{\omega}(\mu\mathcal{L}_{L}^{\mathsf{T}})$ and $\alpha\in \mathsf{T}_\omega (\mu\mathcal{L}_{L}^{\mathsf{T}}(\mathsf{P}))$. There is a restriction on the formulation of the formulas $\mu p.a$, namely, no occurrence of $p$ in $a$ may be in the scope of odd number of negations.\footnote{For a precise definition of the notions \emph{scope} and \emph{occurrence}, we can inductively define a construction tree of a formula, where the children of a node labeled $\nabla\alpha$ are given by the formulas in $Base(\alpha)$. }
\end{dfn}
\begin{rem}\label{subset} For a given formula $a\in\mu\mathcal{L}_{L}^{\mathsf{T}}(\mathsf{P})$, $\mathsf{P}_{a}\subseteq{\mathsf{P}}$ denotes the set of all propositional letters occurring in $a$. Observe that for $\mathsf{Q'}\subseteq \mathsf{Q}\subseteq \mathsf{P}$, we have that $\mu\mathcal{L}_{L}^{\mathsf{T}}(\mathsf{Q'})\subseteq \mu\mathcal{L}_{L}^{\mathsf{T}}(\mathsf{Q})$. This can be proved by induction on the complexity of formulas in $\mu\mathcal{L}_{L}^{\mathsf{T}}(\mathsf{Q'})$.
\end{rem}
Before we turn to the coalgebraic semantics of this language, there are a number of syntactic definitions to be fixed. We start with the definition of subformula.
\begin{dfn}
We will write $b\trianglelefteq a$ if $b$ is a subformula of $a$. Inductively we define the set of $\emph{Sfor}(a)$ of subformulas of $a$ as follows:
\begin{eqnarray}
\emph{Sfor}(p)&:=& \{p\}, \nonumber \\
\emph{Sfor}(\neg a)&:=& \{\neg a\}\cup \emph{Sfor}(a), \nonumber \\
\emph{Sfor}\left( \bigvee A \right) &:=& \{ \bigvee A \} \cup \bigcup_{a\in A}\emph{Sfor}(a), \nonumber \\
\emph{Sfor}(\mu p\cdot a)&:=& \{\mu p\cdot a\}\cup\emph{Sfor}(a), \nonumber \\
\emph{Sfor}\left( \nabla\alpha \right) &:=& \{ \nabla\alpha\} \cup \bigcup_{a\in Base(\alpha)}\emph{Sfor}(a) \nonumber
\end{eqnarray}
The elements of $Base(\alpha)$ will be called the immediate subformulas of $\nabla\alpha$.
\end{dfn}
\begin{dfn} A formula $a\in\mu\mathcal{L}_{L}^{\mathsf{T}}(\mathsf{P})$ is \emph{guarded} if every subformula $\mu p . b$ of $a$ has the property that all occurences of $p$ inside $b$ are within the scope of a $\nabla$.
\end{dfn}
We now introduce the semantics of coalgebraic fixpoint logic. For this purpose we define the notion of a $\mathsf{T}$-model over a set $\mathsf{P}$ of propositional letters.
\begin{dfn}
A $\mathsf{T}$-model $\mathbb{S}=(S,\sigma, V )$ is a $\mathsf{T}$-coalgebra $(S,\sigma)$ together with a valuation $V$ that is a function $V : \mathsf{P}\rightarrow\mathcal{P}(S)$.
\end{dfn}
Using the fixed quasi-functorial lax extension $L$ for the functor $\mathsf{T}$ we can define the semantics for the language $\mu\mathcal{L}_{L}^{\mathsf{T}}(\mathsf{P})$ on $\mathsf{T}$-models, by giving the definition of the satisfaction relation $\Vdash_{\mathbb{S}}:S\rel \mu\mathcal{L}_{L}^{\mathsf{T}}(\mathsf{P})$ for a $\mathsf{T}$-model $\mathbb{S}= (S,\sigma,V)$.
\begin{dfn}\label{semantics} Before going to the definition of the satisfaction relation, we need to fix some notation:
For $X\subseteq S$, $V[p\mapsto X]$ denotes the valuation that is exactly like $V$ apart from mapping $p$ to $X$. We also use ${\llbracket a \rrbracket}_{\mathbb{S}}$ for the extension of formula $a$ in a $\mathsf{T}$-model $\mathbb{S}$: ${\llbracket a \rrbracket}_{\mathbb{S}}:=\{s\in S\mid s\Vdash_{\mathbb{S}} a\}$.Then ${\llbracket a \rrbracket}_{\mathbb{S}[p\mapsto X]}$ denotes the extension of $a$ considering the valuation $V[p\mapsto X]$, instead of $V$.\\ Now we are ready to define the satisfaction relation as follows:
\begin{eqnarray*}
s\Vdash_{\mathbb{S}}p &\text{iff}& s\in V(p)\\
s\Vdash_{\mathbb{S}}\neg a &\text{iff}& \text{not}~s\Vdash_{\mathbb{S}} a\\
s\Vdash_{\mathbb{S}} \bigvee A &\text{iff}& s\Vdash_{\mathbb{S}} a~\text{for} ~\text{some}~a\in A\\
s\Vdash_{\mathbb{S}} \nabla\alpha &\text{iff}& (\sigma(s),\alpha)\in L\Vdash_{\mathbb{S}}\\
s\Vdash_{\mathbb{S}} \mu p.a &\text{iff}& s\in\bigcap\{X\subseteq S\mid {\llbracket a \rrbracket}_{\mathbb{S}[p\mapsto X]}\subseteq X\}.\\
\end{eqnarray*}
\end{dfn}
\begin{rem}
The clauses in Definition~\ref{semantics} are not stated in a correct recursive way. In the recursive clause for the $\nabla$ modality we make use of the unrestricted satisfaction relation $\Vdash_{\mathbb{S}}$ that has yet to be defined. We can only suppose that $\Vdash_{\mathbb{S}}\mid_{S\times Base(\alpha)}$ is already defined. The actual recursive definition is that $s\Vdash_{\mathbb{S}}\nabla\alpha$ iff $(\sigma(s),\alpha)\in L(\Vdash_{\mathbb{S}}\mid_{S\times Base(\alpha)})$. To see why this is equal to the clause given above, see ~\cite[Proposition 6]{jm:yv}.
\end{rem}
Given a valuation $V : \mathsf{P}\rightarrow\mathcal{P}(S)$, one can think of it as a coloring $\gamma_{V}:S\rightarrow\mathcal{P}(\mathsf{P})$ that maps point $s\in S$ to a set of states in $A$: $$\gamma_{V}(s):=\{a\in A\mid s\in V(a)\}.$$ So a $\mathsf{T}$-model $\mathbb{S}=(S,\sigma, V )$ can also be seen as a $\mathcal{P}(\mathsf{P})$-colored $\mathsf{T}$-coalgebra $\hat{\mathbb{S}}=(S, \sigma, \gamma_{V})$. The \emph{projection} of a $\mathcal{P}(\mathsf{P})$-colored $\mathsf{T}$-coalgebra $\mathbb{S}=(S,\sigma, \gamma )$ to a set $\mathsf{Q}\subseteq\mathsf{P}$ is the $\mathcal{P}(\mathsf{Q})$-colored $\mathsf{T}$-coalgebra $\mathbb{S}^\mathsf{Q}=(S, \sigma, \gamma^{\mathsf{Q}})$ where $\gamma^{\mathsf{Q}}: S\rightarrow \mathcal{P}(\mathsf{Q})$, $s\mapsto\gamma(s)\cap \mathsf{Q}.$
\begin{dfn}
Given a set $\mathsf{Q}\subseteq\mathsf{P}$, an $L_{\mathsf{Q}}$-\emph{bisimulation} between two $\mathsf{T}$-models $\mathbb{S}$ and $\mathbb{Y}$ is defined to be an $L_{\mathcal{P}(\mathsf{Q})}$-bisimulation between $\hat{\mathbb{S}}^{\mathsf{Q}}$ and $\hat{\mathbb{Y}}^{\mathsf{Q}}$.
It follows that a relation $R:S\rel Y$ is an $L_{\mathsf{Q}}$-{bisimulation} between $\mathsf{T}$-models $\mathbb{S}=(S, \sigma, V_S)$ and $\mathbb{Y}=(Y,\lambda, V_Y)$ if and only if $R$ is an $L$-bisimulation between $\mathsf{T}$-coalgebras $\mathbb{S}=(S,\sigma)$ and $\mathbb{Y}=(Y, \lambda)$ and $R$ preserves the truth of all propositional letters in $\mathsf{Q}$, that is for all $(s,y)\in R$ we have that for all $ p\in \mathsf{Q}$ $$ s\in V_S(p)~\text{iff}~ y\in V_Y(p).$$
\end{dfn}
From this definition, it is easy to see that for any $\mathsf{Q}'\subseteq\mathsf{Q}$, if a relation $R$ is an $L_{\mathsf{Q}}$-bisimulation between $\mathsf{T}$-models $\mathbb{S}$ and $\mathbb{Y}$, then it is also an $L_{\mathsf{Q}'}$-bisimulation between them.
\begin{dfn}Given a propositional letter $p\in\mathsf{P}$, a relation $R:S\rel S'$ is an \emph{up-to}-$p$ $L_{\mathsf{P}}$-\emph{bisimulation} between two $\mathsf{T}$-models $\mathbb{S}= (S,\sigma,V)$ and $\mathbb{S'}= (S',\sigma',V')$, if it is an $L_{\mathsf{P}\setminus\{p\}}$-bisimulation between $\mathsf{T}$-models $\mathbb{S}$ and $\mathbb{S'}$. We write $ s\leftrightarroweq^{L}_{p} s'$ if $s$ and $s'$ are up-to-$p$ $L_{\mathsf{P}}$-bisimilar, that is where we disregard the proposition letter $p$.
\end{dfn}
Now we are going to look at the expressive power of $\mu\mathcal{L}^{\mathsf{T}}_L(\mathsf{P})$ with respect to states in $\mathsf{T}$-models. For this, we start with a definition.
\begin{dfn}
Two states $s$ in $\textsf{T}$-model $\mathbb{S}=(S,\sigma,V)$ and $s'$ in $\textsf{T}$-model $\mathbb{S'}=(S',\sigma',V')$ are called \emph{equivalent} for formulas in $\mu \mathcal{L}_L^{\textsf{T}}(\textsf{P})$ if $s\Vdash_{\mathbb{S}} a$ iff $s'\Vdash_{\mathbb{S}'} a$, for all $a\in \mu \mathcal{L}_L^{\textsf{T}}(\textsf{P})$.
\end{dfn}
An important property of our coalgebraic fixpoint logic is that truth is bisimulation invariant. This fact is given by the following proposition.
\begin{prop}
Given a state $s$ in a $\textsf{T}$-model $\mathbb{S}=(S,\sigma,V)$ and a state $s'$ in a $\textsf{T}$-model $\mathbb{S'}=(S',\sigma',V')$, if $s$ and $s'$ are $L_\textsf{P}$-bisimilar then $s$ and $s'$ are equivalent for formulas in $\mu \mathcal{L}_L^{\textsf{T}}(\textsf{P})$.
\end{prop}
For the proof of this proposition we refer to ~\cite[Proposition 5.14]{yv}, ~\cite[Proposition 4.11]{jm} and the fact that lax extensions are monotone.

Now we are ready to state the last semantic result we will need through out this paper.
\begin{prop} Each formula in $\mu\mathcal{L}_{L}^{\mathsf{T}}(\mathsf{P})$ can be transformed into an equivalent guarded formula in $\mu\mathcal{L}_{L}^{\mathsf{T}}(\mathsf{P})$.
\end{prop}
It can be proved by induction on the complexity of formulas, see~\cite[Proposition 5.15]{yv}
\begin{conv} Throughout this paper we always assume $\mu\mathcal{L}_{L}^{\mathsf{T}}(\mathsf{P})$-formulas to be guarded.
\end{conv}
\subsection{Coalgebraic Automata}
Coalgebraic automata are supposed to operate on pointed coalgebras. Basically, the idea is that an initialized $\mathsf{T}$-automaton will either \emph{accept} or \emph{reject} a given pointed $\mathsf{T}$-coalgebra. In the following section, we will recall the basic definitions from coalgebraic automata theory.
\begin{dfn}\label{automata}
Let $\mathsf{T}:\mathsf{Set}\rightarrow\mathsf{Set}$ be a set functor. A (non-deterministic) $\mathsf{T}$-automaton over a color set $C$ is a triple $\mathbb{A}=(A,\Delta,\Omega)$, with $A$ some finite set (of states), $\Delta:A\times C\rightarrow\mathcal{P}(\mathsf{T}A)$ the \emph{transition function} and $\Omega:A\rightarrow\omega$ a \emph{parity map}. The \emph{initialized} version of $\mathbb{A}$ is the pair $(\mathbb{A},a)$ consisting of an automaton $\mathbb{A}$ together with an
element $a\in A$, which we call it's \emph{initial} state.
\end{dfn}
The acceptance condition for $\mathsf{T}$-automata is formulated in terms of a parity game\cite {ck:yv}. The acceptance game $\mathcal{G}(\mathbb{S},\mathbb{A})$ between initialized automaton $(\mathbb{A},a_I)$ and a pointed coalgebra $(\mathbb{S},s_I)$ is given by the Table 1.
The game is played by two players: \'{E}loise ($\exists$) and Ab\'{e}lard ($\forall$). A \emph{match} of the game is a (finite or infinite) sequence of positions which is given by the two players moving from one position to another according to the rules of Table 1.
Let use now give the formal definition of acceptance game.

\begin{dfn} Let $(\mathbb{A},a_I)$ be an initialized $\mathsf{T}$-automaton over the color set $C$. Furthermore let $(\mathbb{S},s_I)=(S,\sigma,\gamma, s_I)$ be a pointed $C$-colored $\mathsf{T}$-coalgebra. Then the \emph{acceptance game} $\mathcal{G}(\mathbb{S},\mathbb{A})$ is given by the following table:
\begin{table}[h]
\centering
\begin{tabular}{|l|c|l|l|}
\hline
Position & Player & Admissible moves & Priority \\ \hline
$(s,a)\in S\times A$ & $\exists$ & $(\sigma(s),\phi)~s.t. ~\phi\in\Delta(a,\gamma(s))$ & $\Omega(a)$\\
$(\sigma(s),\phi)\in\mathsf{T}S\times\mathsf{T}A$ & $\exists$ & $\{Z:S\rel A\mid(\sigma(s),\phi)\in LZ$ & 0 \\
$Z\subseteq S\times A$& $\forall$ & $Z$ & 0 \\
\hline
\end{tabular}
\caption{Acceptance game for $\mathsf{T}$-automaton}
\end{table}

Positions of the form $(s,a)\in S\times A$ will be called \emph{basic positions} of the game. A partial play of the game of the form $(s,a)(\sigma(s),\phi)Z(t,b)$ with $(s,a)\in S\times A$, $(\sigma(s),\phi)\in\mathsf{T}S\times\mathsf{T}A$, $Z:S\rel A$ and $(t,b)\in Z$ will be called a \emph{round} of the play. A $\emph{positional}$ or $\emph{history free strategy}$ for $\exists$ is a pair of functions $$(\Phi:S\times{A}\rightarrow\mathsf{T}A, Z:S\times A\rightarrow\mathcal{P}(S\times A)).$$
Such a strategy is $\emph{legitimate}$ if at any position, it maps the position to an admissible next position as given by Table 1. A legitimate strategy is \emph{winning} for $\exists$ from a position in the game, if it guarantees $\exists$ to win any match starting from that position, no matter how $\forall$ plays. A position starting from which $\exists$ has a winning strategy is called a \emph{winning position} for $\exists$ . The set of all winning positions for $\exists$ in $\mathcal{G}(\mathbb{S},\mathbb{A})$ is denoted by $\mbox{Win}_{\exists}(\mathbb{S},\mathbb{A})$ or shortly by $\mbox{Win}_{\exists}$.
A history-free strategy $(\Phi, Z)$ initialized at $(s_I, b)\in S\times A$ is called $\emph{scattered}$ if the relation $$\{(s_I, b)\}\cup\bigcup\{Z_{s,a}\subseteq S\times{A}\mid(s,a)\in\mbox{Win}_{\exists}\}$$ is functional. Finally we say that initialized $\mathsf{T}$-automaton $(\mathbb A,a_I)$ $\emph{accepts}$ $(\mathbb{S},s_I)$ if $\exists$ has a winning strategy in the game $\mathcal{G}(\mathbb {A}, \mathbb{S})$ initialized at position $(s_I, a_I)$. If $\exists$ has a scattered winning strategy starting from $(s_I, a_I)$, we will say $(\mathbb A,a_I)$ \emph{strongly accepts} $(\mathbb{S},s_I)$.
\end{dfn}
\begin{dfn}
For every initialized $\mathsf{T}$-automaton $(\mathbb{A},a_I)$ over some color set $C$, $L(\mathbb{A}, a_I)$, the \emph{recognizable language} of $(\mathbb{A},a_I)$, is the class of all pointed $C$-colored $\mathsf{T}$-coalgebras that are accepted by $(\mathbb{A},a_I)$. We call two initialized $\mathsf{T}$-automata $(\mathbb{A},a_I)$ and $(\mathbb{A'},a'_I)$ over set $C$ \emph{equivalent} iff $L(\mathbb{A},a_I)=L(\mathbb{A'},a'_I)$.
\end{dfn}
\subsection{Logic vs. Automata}
\begin{prop}\label{Prop4}
There exists an effective procedure to transform a formula $b\in\mu \mathcal{L}_L^{\textsf{T}}(\textsf{P})$ to an initialized $\textsf{T}$-automaton $(\mathbb{A}_b,a_b)$ over the set $C=\mathcal{P}(\textsf{P})$ such that for every $C$-colored $\textsf{T}$-coalgebra $(\mathbb{S},s)$:
$$(\mathbb{S},s)\Vdash_{\mathbb{S}} b \text{ iff }(\mathbb{A}_b,a_b) \text{ accepts } (\mathbb{S},s).$$
\end{prop}
Conversley, there is an effective procedure to construct a $\mu \mathcal{L}_L^{\textsf{T}}(\textsf{P})$-formula $a_{\mathbb{A}}$ for a given initialized $\textsf{T}$-automaton $(\mathbb{A},a_I)$ such that $a_{\mathbb{A}}$ holds precisely at those pointed $\textsf{T}$-coalgebras that are accepted by $(\mathbb{A},a_I)$. This fact is given by the following proposition:
\begin{prop}\label{prop5}
There exists an effective procedure transforming an initialize $\textsf{T}$-automaton $(\mathbb{A}, a_I)$ to an equivalent $\mu \mathcal{L}_L^{\textsf{T}}(\textsf{P})$-formula $a_{\mathbb{A}}$.
\end{prop}
\section{Automata are Closed under Projection}

This section is devoted to proof of the main technical result of our paper i.e.; closure under projection.

\begin{dfn} Let $\mathbb{A}=(A, \Delta, \Omega)$ be a $\textsf{T}$-automaton over color set $C$. We call a state $a\in A$ a \emph{true state} of $\mathbb{A}$ if $\Omega(a)$ is even and $\Delta(a,c)=\textsf{T}(\{a\})$. We will standardly use the notation $a_{\top}$ to refer to a true state.
Given $(a,c)\in A\times C$ we call $\phi\in\Delta(a,c)$ a \emph{satisfiable element} of $\mathbb{A}$ if there is a witnessing $\textsf{T}$-coalgebra $(\mathbb{Q}_{\phi},\rho, \gamma_Q)$, $\tau\in\mathsf{T}Q$ and a relation $Z_\phi:Q\rel A$ such that $(\tau, \phi)\in LZ_\phi$ and $Z_{\phi}\subseteq\mbox{Win}_{\exists}(\mathbb{Q},\mathbb{A})$. Finally we call a $\mathsf{T}$-automaton $\mathbb{A}$  \emph{totally satisfiable} whenever for all $(a,c)\in A\times C$ and $\phi\in\Delta(a,c)$, $\phi$ is satisfiable.
\end{dfn}

The following proposition states that without loss of generality we can always assume that an initialized $\mathsf{T}$-automaton $(\mathbb{A},a_I)$  is totally satisfiable and has a true state. Furthermore, we may always assume that there exists a witnessing $\mathsf{T}$-coalgebra $\mathbb{Q}$ that works for all $(a,c)\in A\times C$ and $\phi\in \Delta(a,c)$. 
\begin{prop}\label{factaut}
For any initialized $\mathsf{T}$-automaton $(\mathbb{A},a_I)$ over set color $C$ we have that:
\begin{itemize}
\item[(1)]There is an equivalent initialized $\mathsf{T}$-automaton $(\mathbb{A'},a_I)$ such that $\mathbb{A}'$ has a true state.
\item[(2)]There exists a totally satisfiable initialized $\mathsf{T}$-automaton $(\mathbb{A'},a'_I)$ which is equivalent to $(\mathbb{A},a_I)$.
\item[(3)]For every totally satisfiable automaton $(\mathbb{A},a_I)$ there is a $C$-colored witnessing coalgebra $\mathbb{Q}=(Q,\rho,\gamma_{Q})$ and a relation $Y:Q\rel A$ such that for all $(a,c)\in A\times C$ and $\phi\in \Delta (a,c)$, there is a $\tau\in \mathsf{T}Q$ such that $(\tau, \phi)\in LY$ and $Y\subseteq\mbox{Win}_{\exists}(\mathbb{Q},\mathbb{A})$.
\end{itemize}
\end{prop}

\begin{proof}
\begin{itemize}
\item[(1)] Define $(\mathbb{A}', a'_I):=(A\cup\{a_{\top}\}, \Delta', \Omega', a'_I)$ such that $a'_I=a_I$ and for all $a\in A$, $\Delta(a)=\Delta'(a)$ and $\Omega(a)=\Omega'(a)$. For $(a_{\top},c)\in A\times C$ define $\Delta'(a_{\top},c):=\mathsf{T}(\{a_{\top}\})$ and $\Omega' (a_{\top}):=0$. Since it is not difficult to check the equivalence of these automata, we leave it for the reader.

\item[(2)] We will define $(\mathbb{A}',a'_I)$ over $C$ by just removing the unsatisfiable elements of any $\Delta(a,c)$:
$$(\mathbb{A}',a'_I)=(A,\Delta',\Omega,a_I),$$
where $\Delta'(a,c)=\{\phi\in \Delta(a,c)\mid \phi ~\text{is a satisfiable element}\}$. 

$(\mathbb{A}', a_I)$ and $(\mathbb{A}, a_I)$ are equivalent since $\exists$ will never go through unsatisfiable elements in winning plays.

\item[(3)] Take the coproduct of all witnessing coalgebra $\mathbb{Q}_\phi$ for all $\phi\in\Delta(a,c)$ for every $(a,c)\in A\times C$. The relation $Y$ is the union of all $Y_\phi$.
\end{itemize}
\end{proof}

Now we will state the main technical result of this paper. Theorem~\ref{thm1} is a generalization of ~\cite[Proposition 5.9]{ck:yv}, where the same result is proved for the weak-pullback preserving functors. In the following theorem we will generalize the proposition to the class of all functors with a quasi-functorial lax extension that preserves diagonals.

\begin{thm}[\textbf{Closure under projection}]\label{thm1}
 Given an initialized $\mathsf{T}$-automaton $(\mathbb{A},a_I)$ over a color set $\mathcal{P}(\mathsf{P})$ and an element $p\in\mathsf{P}$, then there exists an initialized $\mathsf{T}$-automaton $(\exists_{p}.\mathbb{A},a)$ over color set $\mathcal{P}(\mathsf{P}\setminus\{p\})$ such that:
\begin{eqnarray}\label{L}
(\mathbb{S},s_I)\in L(\exists_{p}.\mathbb{A},a))~\text{iff}~(\mathbb{\overline{S}},\overline{s}_I)\in L(\mathbb{A},a)~\text{for some}~(\mathbb{\overline{S}},\overline{s}_I)~\text{with}~\mathbb{S},s_I\leftrightarroweq_{p}^{L}\mathbb{\overline{S}},\overline{s}_I.
\end{eqnarray}
\end{thm}

\begin{proof}
Given $(\mathbb{A}, a)$ over color set $\mathcal{P}(\mathsf{P})$, we define the initialized $\mathsf{T}$-automaton $(\exists_{p}.\mathbb{A},a)$ over color set $\mathcal{P}(\mathsf{P}\setminus\{p\})$ as the following automaton:
$$(\exists_{p}.\mathbb{A},a):=(A, \Delta_{p}, \Omega, a),$$
where $\Delta_{p}: A\times\mathcal{P}(\mathsf{P}\setminus\{p\})\rightarrow \mathcal{P}\mathsf{T}A$, $(a, c)\mapsto\Delta(a, c)\cup\Delta(a, c\cup\{p\})$.

 In order to show that (\ref{L}) holds, we start with the direction from right to left:\\
$(\Longleftarrow)$ We claim that if the initialized $\mathsf{T}$-automaton $(\mathbb{A}, a_I)$ accepts a $\mathcal{P}(\mathsf{P})$-colored $\mathsf{T}$-coalgebra $(\mathbb{S'},s'_I)$, then $(\exists_{p}.\mathbb{A},a)$ accepts $(\mathbb{S'}_p, s'_I)$; the projection of $(\mathbb{S'},s'_I)$ to the set $\mathsf{P}\setminus\{p\}$. Proof of the claim is straightforward, since all legitimate moves of $\exists$ in the game $\mathcal{G}(\mathbb{A},\mathbb{S'})$ are still legitimate moves of $\exists$ in the game $\mathcal{G}(\exists_{p}.\mathbb{A},\mathbb{S'}_p)$.\\
$(\Longrightarrow)$ Let us assume that $(\exists_{p}.\mathbb{A},a)$ accepts $\mathcal{P}(\mathsf{P}\setminus\{p\})$-colored $\mathsf{T}$-coalgebra $(\mathbb{S}, s_I)=(S,\sigma,\gamma, s_I)$, we will define a $\mathcal{P}(\mathsf{P})$-colored coalgebra $(\overline{S},\overline{s}_I)$ such that it satisfies (\ref{L}).\\
From Proposition \ref{factaut} it follows that $(\mathbb{A},a_I)$ is totally satisfiable and has a true state. In addition we get a $\mathcal{P}(\mathsf{P})$-colored witnessing coalgebra $\mathbb{Q}=(Q,\rho,\gamma_{Q})$. In the following we will give the construction of $(\overline{S},\overline{s}_I)$ using $(\mathbb{S},s_I)$ and $\mathbb{Q}$.

We put $\overline{S}:=(S\times A)\uplus Q$ and in order to define the coalgebra structure $\overline{\sigma}:\overline{S}\rightarrow\mathsf{T}\overline{S}$ we distinguish the following cases:

\begin{itemize}
\item[(1)]$q\in Q$, define $\overline{\sigma}(q):=\rho(q)$,
\item[(2)]$(s,a)\in S\times A$ and $(s,a)\notin\mbox{Win}_{\exists}(\mathbb{S},\exists_{p}.\mathbb{A})$, define $\overline{\sigma}(s,a):=\mathsf{T}\kappa_{a}(\sigma(s))$, where $\kappa_{a}:S\rightarrow S\times A$, $s\mapsto (s,a)$, 
\item[(3)]$(s,a)\in S\times A$ and $(s,a)\in\mbox{Win}_{\exists}(\mathbb{S},\exists_{p}.\mathbb{A})$. In this case from $\exists$'s winning strategy in $\mathcal{G}(\mathbb{S},\exists_{p}.\mathbb{A})@(s,a)$ we get a $\phi_{s,a}\in\Delta(a,\gamma(s))$ and a relation $Z_{s,a}:S\rel A$ such that $Z_{s,a}\subseteq\mbox{Win}_{\exists}(\mathbb{S},\exists_{p}.\mathbb{A})$ and $(\sigma(s),\phi_{s,a})\in LZ_{s,a}$. We extend relation $Z_{s,a}$ to the relation $Z'_{s,a}:S\rel A$ as follows:
$$Z'_{s,a}:= Z_{s,a}\cup \{(t,a_{\top})\mid t\notin Dom(Z_{s,a})\}.$$
Considering the projection maps $\pi_{1}:Z'_{s,a}\rightarrow S$ and $\pi_{2}:Z'_{s,a}\rightarrow A$, we get that $Z'_{s,a}=\pi_{1}^{\circ};\pi_{2}$
\begin{clm*}[\textbf{1}] $(\sigma(s),\phi_{s,a})\in L(\pi_{1}^{\circ};i);L(\pi_{2}\cup Y)$, where $Y$ is given by Proposition~\ref{factaut}(3), from totally satisfiability of $(\mathbb{A}, a_I)$, and $i:Z\hookrightarrow Z\uplus Q$ is inclusion map.
\end{clm*}
\emph{Proof of Claim (1)}.
\begin{itemize}
\item[(i)]$\sigma(s)\in Dom(L(\pi_{1}^{\circ};i))$ by Proposition~\ref{diag}~$(2)$ and fullness of $Z'_{s,a}$ on $S$.
\item[(ii)]$\phi_{s,a}\in Rng (L(\pi_{2}\cup Y))$, because $\phi_{s,a}\in Rng(LY)$ (by definition of $Y$) and $Rng (LY)\subseteq Rng (L(\pi_{2}\cup Y))$
\item[(iii)]$(\sigma(s),\phi_{s,a})\in LZ'_{s,a}= L(\pi_{1}^{\circ};\pi_{2})$, since $(\sigma(s),\phi_{s,a})\in LZ_{s,a}$ and $LZ_{s,a}\subseteq LZ'_{s,a}$.

Hence by quasi-functoriality of $L$ and $\pi_{1}^{\circ};\pi_2\subseteq (\pi_{1}^{\circ};i);(\pi_2\cup Y)$ we get that $$(\sigma(s),\phi_{s,a})\in L\pi_{1}^{\circ};L(\pi_{2}\cup Y).$$
\end{itemize}
We may take $\overline{\sigma}(s,a)\in\mathsf{T}(Z'_{s,a}\uplus Q)\subseteq\mathsf{T}(S\times A)\subseteq\mathsf{T}{\overline{S}}$ such that 
$$(\sigma(s),\overline{\sigma}(s,a))\in L\pi_{1}^{\circ}~\text{and}~(\overline{\sigma}(s,a),\phi_{s,a})\in L(\pi_{2}\cup Y).$$
\end{itemize}
In order to complete the definition of $\mathcal{P}(\mathsf{P})$-colored pointed coalgebra $(\overline{\mathbb{S}}, (s_I,a_I)$ we have to introduce a coloring $\overline{\gamma}:\overline{S}\rightarrow \mathcal{P}(\mathsf{P})$. We do so by distinguishing the following cases:
\begin{itemize}
\item[(1)]$q\in Q$ define $\overline{\gamma}(q):=\gamma_{Q}(q)$,
\item[(2)]$(s,a)\in S\times A$ and $(s,a)\notin\mbox{Win}_{\exists}(\mathbb{S},\exists_{p}.\mathbb{A})$, define $\overline{\gamma}(s,a):=\gamma(s)$,
\item[(3)]$(s,a)\in S\times A$ and $(s,a)\in\mbox{Win}_{\exists}(\mathbb{S},\exists_{p}.\mathbb{A})$. In this case we will define $\overline{\gamma}(s,a)$ by considering the choice of $\exists$ at $(s,a)$. Since $(s,a)$ is a winnig position for $\exists$, she picks an element $\phi_{s,a}\in\Delta_{p}(a, \gamma(s))$. But from the definition of $\Delta_p$ we know that $$\phi_{s,a}\in\Delta_{p}(a, \gamma(s))=\Delta(a, \gamma(s))\cup\Delta(a, \gamma(s)\cup\{p\}).$$
We define $\overline{\gamma}(s,a):=\gamma(s)\cup\{p\}$ if $\phi_{s,a}\in\Delta(a, \gamma(s)\cup\{p\})$, otherwise we put $\overline{\gamma}(s,a):=\gamma(s)$.
\end{itemize}
\begin{clm*}[\textbf{2}] ~$\mathbb{S},s_I\leftrightarroweq_{p}^{L}\mathbb{\overline{S}},(s_I,a_I)$.
\end{clm*}
\begin{flushleft}
\emph{Proof of claim (2)}.
\end{flushleft} We will show that the graph of partial map $\pi_S:\overline{S}\rightarrow S$ is an up-to-$p$ bisimulation between $\mathbb{S},s_I$ and $\mathbb{S},s_I$, so we need to prove the following:
$$(\overline{\sigma}(\overline{s}),\sigma(s))\in L\pi_{S} ~\text{and}~\overline{\gamma}(\overline{s})\setminus\{p\}=\gamma(s)~\text{whenever}~(\overline{s},s)\in\pi_S.$$

We have two cases:
\begin{itemize}
\item[(i)]$(s,a)\in S\times A$ and $(s,a)\notin\mbox{Win}_{\exists}(\mathbb{S},\exists_{p}.\mathbb{A})$.
In this case the statement holds since from the definition of $\overline{\sigma}$ we have that:
\begin{eqnarray*}
(\overline{\sigma}(s,a),\sigma(s)) &=& (\mathsf{T}\kappa_{a}(\sigma(s)),\sigma(s))\\
&\in & (\textsf{T}\kappa_{a})^{\circ}\\
&=& L(\kappa_{a}^{\circ})\\
&\subseteq & L\pi_{S}
\end{eqnarray*}

and from the definition of $\overline{\gamma}$ we have that in this case $\overline{\gamma}(\overline{s})=\gamma(s)$.

\item[(ii)]$(s,a)\in S\times A$ and $(s,a)\in \mbox{Win}_{\exists}(\mathbb{S},\exists_{p}.\mathbb{A})$. In this case by the definition of $\overline{\sigma}$ we get that: 
$$(\overline{\sigma}(s,a),\sigma(s))\in L\pi_{1}\subseteq\L\pi_{S},$$
and again from the definition of $\overline{\gamma}$ it is clear that $\overline{\gamma}(\overline{s})\setminus\{p\}=\gamma(s)$.
\end{itemize}

\begin{clm*}[\textbf{3}] ~$((s_I,a_I),a_I)\in\mbox{Win}_{\exists}(\overline{\mathbb{S}},\mathbb{A})$.
\end{clm*}
\begin{flushleft}
\begin{flushleft}
\emph{{Proof of claim (3)}}.
\end{flushleft} Let $(\Phi, Z)$ be a winning strategy for $\exists$ in $\mathcal{G}(\mathbb{S}, \exists_{p}.\mathbb{A})@(s_I,a_I)$ and $(\Psi, Y)$ be $\exists$'s strategy in $\mathcal{G}(\mathbb{Q}, \mathbb{A})$. Define $\exists$'s strategy in $\mathcal{G}(\overline{\mathbb{S}},\mathbb{A})$ as follows:
\end{flushleft}

\begin{eqnarray*}
\overline{\Phi} : \overline{S}\times A &\rightarrow& \textsf{T} A \\ ((s,b),a)  &\mapsto & \phi_{s,a}\\ (q,a) &\mapsto & \psi_{q,a}
\end{eqnarray*}

\begin{eqnarray*}
\overline{Z} : \overline{S}\times A &\rightarrow&\mathcal{P}(\overline{S}\times A) \\ ((s,b),a) &\mapsto & \pi_{2}\cup Y\\ (q,a) &\mapsto & Y
\end{eqnarray*}

where $Y$ is given by $(\Psi, Y)$ and $\pi_{2}:Z'_{s,a}\rightarrow A$ such that $Z'_{s,a}$ is the extension of $Z_{s,a}$ given by $(\Phi, Z)$ at position $(s,a)$.\\
\begin{clm*}[\textbf{3a}] For the following types of positions in $\mathcal{G}(\overline{\mathbb{S}}, \mathbb{A})$, the given strategy $(\overline{\Phi}, \overline{Z})$ provides legitimate moves for $\exists$:
\begin{itemize}
\item[(i)] $(q,a)\in \overline{S}\times A$ and $(q,a)\in \mbox{Win}_{\exists}(\mathbb{Q},\mathbb{A})$,
\item[(ii)] $((s,a),a)\in \overline{S}\times A$ and $(s,a)\in \mbox{Win}_{\exists}(\mathbb{S},\exists_{p}.\mathbb{A})$
\end{itemize}
\end{clm*}
\begin{flushleft}
\emph{Proof of Claim (3a)}. 
\end{flushleft}
\begin{itemize}
\item[(i)] It is clear since $\overline{\sigma}(q)=\rho(q)$ and at this position $\exists$ plays her winning strategy in $\mathcal{G}(\mathbb{Q},\mathbb{A})$.

\item[(ii)] We need to show that: $$(\overline{\sigma}(s,a),\overline{\phi}_{s,a})\in L\overline{Z}_{s,a},$$

but this is simply the case from the definition of $\overline{\sigma}$:
$$(\overline{\sigma}(s,a),\overline{\phi}_{s,a})=(\overline{\sigma}(s,a),\phi_{s,a})\in L(\pi_{2}\cup Y)$$
\end{itemize}

This finishes the proof of claim $3a$.\\
\begin{clm*}[\textbf{3b}] ~$(\overline{\Phi},\overline{Z})$ guarantees $\exists$ to win any match of $\mathcal{G}(\overline{\mathbb{S}}, \mathbb{A})$ starting from $((s_I,a_I),a_I)$.
\end{clm*}
\begin{flushleft}
\emph{Proof of Claim (3b)}.
\end{flushleft} To provide this, consider an arbitrary match which conforms the strategy $(\overline{\Phi},\overline{Z})$. From the definition of this strategy, it is clear that at each round of the match $\forall$ may have three different types of positions to choose from: elements of the form $(q,a)$, elements of the form $((s,a),a)$ and elements of the form $((t,a_{\top}),a_{\top})$. So to check whether $(\overline{\Phi},\overline{Z})$ is indeed a winning strategy for $\exists$ we may distinguish the following matches:

\begin{itemize}
\item[(i)] At some stage $\forall$ chooses an element $(q,a)\in Y$. From this moment on, there is no way to  go through the states of $\mathbb{S}$  and since $Y\subseteq\mbox{Win}_{\exists}(\mathbb{Q},\mathbb{A})$, $\exists$ plays her winning strategy in  $\mathcal{G}(\mathbb{Q}, \mathbb{A})@(q,a)$ and wins the match.
\item[(ii)] $\forall$ always picks an element of the form $((s,a),a)$. In this case the match will never go through the states of $\mathbb{Q}$ and for any $(\overline{\Phi},\overline{Z})$-conform match starting from $((s_I,a_I),a_I)$
$$((s_I,a_I),a_I)((s_1,a_1),a_1)((s_2,a_2),a_2)\dots$$
in $\mathcal{G}(\overline{\mathbb{S}}, \mathbb{A})$, the corresponding match
$$(s_I,a_I) (s_1,a_1) (s_2,a_2) \ldots$$
in $\mathcal{G}(\mathbb{S}, \mathbb{A})$ conforms $(\Phi,Z)$. And similar as in the previous argument, since we assumed $(\Phi,Z)$ to be a winning strategy for $\exists$, $(\overline{\Phi},\overline{Z})$ is also a winning strategy for her.
\item[(iii)] At some stage $\forall$ picks an element 
$((t,a_{\top}),a_{\top})$. Then from the assumption that $a_{\top}$ is a true state of the automaton $\mathbb{A}$, it follows that playing the strategy $(\overline{\Phi},\overline{Z})$ from $((t,a_{\top}),a_{\top})$ is a win for $\exists$.
\end{itemize}

This finishes the proof of claim $(3)$, and so the proof of Theorem\ref{thm1}.

\end{proof}

\section{Uniform Interpolation for $\mu\mathcal{L}^{\mathsf{T}}_L$}
In the following section we will prove the main theorem of this paper, uniform interpolation for $\mu \mathcal{L}_L^{\textsf{T}}(\textsf{P})$. In order to do this we  introduce the notion of a bisimulation quantifier and show that bisimulation quatifiers are definable in the language $\mu \mathcal{L}_L^{\textsf{T}}(\textsf{P})$. Our proof follows the proof in~\cite{ls} which shows a similar result for monotone modal logic.
\begin{dfn}
Define the relation of logical consequence $$\vDash : \mu \mathcal{L}_L^{\textsf{T}}(\textsf{P})\rel \mu \mathcal{L}_L^{\textsf{T}}(\textsf{P})$$ by $a \vDash a'$ iff $s\Vdash_{\mathbb{S}}a$ implies $s\Vdash_{\mathbb{S}}a'$ for all states $s$ in any $\textsf{T}$-model $\mathbb{S}$.
\end{dfn}
\begin{dfn}
Given a propositional letter $p$, the bisimulation quantifier $\exists{p}$ is an operator with the following semantics for any formula $b\in\mu\mathcal{L}_{L}^{\mathsf{T}}(\mathsf{P})$:
\begin{eqnarray}\label{semantic}
\mathbb{S},s\Vdash\exists{p}.b\quad\textmd{iff}\quad\mathbb{S'},s'\Vdash b,~\textmd{for some}~\mathbb{S'},s'~\textmd{with}~\mathbb{S},s\leftrightarroweq_{p}^{L}\mathbb{S'},s'.
\end{eqnarray}
\end{dfn}

The following proposition shows that the bisimulation quantifier is definable in the language $\mu\mathcal{L}_{L}^{\mathsf{T}}(\mathsf{P})$

\begin{prop}
Given a proposition letter $p$, there is a map\\ $\exists{p}:\mu\mathcal{L}_{L}^{\mathsf{T}}(\mathsf{P})\longrightarrow\mu\mathcal{L}_{L}^{\mathsf{T}}(\mathsf{P})$ such that for any formula $b\in\mu\mathcal{L}_{L}^{\mathsf{T}}(\mathsf{P})$, the semantics of $\exists{p}.b$ is given by ~(\ref{semantic}).
\end{prop}

\begin{proof}
For the proof of this proposition we will use results from section $3$ and section $4$ and define the map  $\exists{p}:\mu\mathcal{L}_{L}^{\mathsf{T}}(\mathsf{P})\longrightarrow\mu\mathcal{L}_{L}^{\mathsf{T}}(\mathsf{P})$ as follows:

Take a formula $b\in\mu\mathcal{L}_{L}^{\mathsf{T}}(\mathsf{P}) $, by Proposition~\ref{Prop4} we can transform it to an equivalent initialized $\mathsf{T}$-automaton $(\mathbb{A}_{b}, a_b)$. From the Theorem~\ref{thm1} we have an initialized $\mathsf{T}$-automaton $(\exists_{p}.\mathbb{A}_{b}, a_b)$ such that:
\begin{eqnarray*}
(\exists_{p}.\mathbb{A}_{b}, a_b)~\text{accepts}~(\mathbb{S},s)~\text{iff}~(\mathbb{A}_b, a_b)~\text{accepts}~(\mathbb{S'},s')~\text{for some}~(\mathbb{S'},s')~\text{with}~\mathbb{S},s\leftrightarroweq_{p}^{L}\mathbb{S'},s'.
\end{eqnarray*}

Now by the Proposition~\ref{prop5} we can transform the initialized $\mathsf{T}$-automaton $(\exists_{p}.\mathbb{A}_{b}, a_b)$ to an equivalent formula $a_{(\exists_{p}.\mathbb{A}_{b})}$ and put $\exists{p}. b:=a_{(\exists_{p}.\mathbb{A}_{b})}$. 
It is easy to show that: 

$$\mathbb{S},s\Vdash a_{(\exists_{p}.\mathbb{A}_{b})}\quad\textmd{iff}\quad\mathbb{S'},s'\Vdash b,~\textmd{for some}~\mathbb{S'},s'~\textmd{with}~\mathbb{S},s\leftrightarroweq_{p}^{L}\mathbb{S'},s',
$$

since we have:

\begin{align*}
\mathbb{S},s\Vdash a_{(\exists_{p}.\mathbb{A}_{b})}
&~\text{iff}~(\mathbb{S},s)\in L(\exists_{p}.\mathbb{A}_{b}, a_b)\\
&~\text{iff}~(\mathbb{S'},s')\in L(\mathbb{A}_b, a_b) ~\text{for some}~(\mathbb{S'},s')~\text{with}~\mathbb{S},s\leftrightarroweq_{p}^{L}\mathbb{S'},s' \\
&~\text{iff}~\mathbb{S'},s'\Vdash b,~\textmd{for some}~\mathbb{S'},s'~\textmd{with}~\mathbb{S},s\leftrightarroweq_{p}^{L}\mathbb{S'},s'.
\end{align*}

\end{proof}

\begin{rem} The function $\exists{p}$ removes all occurances of the propositional letter $p$ from its argument. This means that it restricts to a mapping\newline $\exists{p}:\mu\mathcal{L}_{L}^{\mathsf{T}}(\mathsf{P})\longrightarrow\mu\mathcal{L}_{L}^{\mathsf{T}}(\mathsf{P}\setminus\{p\})$.
\end{rem}

 Now we are ready to prove the uniform interpolation theorem:

\begin{thm}[{\textbf{Uniform Interpolation for}}$ ~\mu\mathcal{L}_{L}^{\mathsf{T}}$]\label{uniform}
For any formula $a\in\mu \mathcal{L}_L^{\textsf{T}}(\mathsf{P})$ and any set $\mathsf{Q}\subseteq \mathsf{P}_{a}$ of propositional letters, there is a formula $a_\mathsf{Q}\in\mu \mathcal{L}_L^{\textsf{T}}(\mathsf{Q})$, effectively constructable from $a$, such that for every formula $b\in\mu \mathcal{L}_L^{\textsf{T}}(\mathsf{P})$ with $\mathsf{P}_a\cap \mathsf{P}_b\subseteq \mathsf{Q}$, we have that $$a\vDash b\quad\text {iff}\quad a_{\mathsf{Q}}\vDash b.$$
\end{thm}
\begin{proof}
Let $\{p_{0}, p_{1},...,p_{n-1}\}=\mathsf{P}_a\setminus \mathsf{Q}$. Then set $$a_{\mathsf{Q}}:=\exists p_0.\exists p_1. \ldots\exists p_{n-1}. a.$$
In order to check that $a\vDash b$ iff $a_{\mathsf{Q}}\vDash b$, first assume that $a\vDash b$. To prove that $a_{\mathsf{Q}}\vDash b$ take a pointed $\textsf{T}$-model $(\mathbb{S}_{0}, s_0)$ with $s_0\Vdash_{\mathbb{S}_0} a_{\mathsf{Q}}$. By the semantics of the bisimulation quantifiers we get states $s_i$ in $\textsf{T}$-models $\mathbb{S}_i$ for $i=1,2,\ldots,n$ such that $s_i\leftrightarroweq_{p_i} s_{i+1}$ for $i=0,...,n$ and $s_n\Vdash_{\mathbb{S}_n} a$. From the later fact it follows that $s_n\Vdash_{\mathbb{S}_n} b$ since we have assumed $a\vDash b$.
Because each of the witnessing up-to-$p_i$ $L_{\textsf{P}}$-bisimulations for $i=0,1,\ldots,n-1$ is also an $L_{\textsf{P}\setminus\{p_0,p_1,\ldots, p_{n-1}\}}$-bisimulation, we can compose them and obtain an $L_{\textsf{P}\setminus\{p_0,p_1,\ldots, p_{n-1}\}}$-bisimulation between $s_0$ and $s_n$. Since $\mathsf{P}_b\subseteq \textsf{P}\setminus\{p_0,p_1,\ldots, p_{n-1}\}$ we get $s_0\Vdash_{\mathbb{S}_0}b$.

For the other direction we show that $a\vDash a_{\mathsf{Q}}$. Then $a\vDash b$ follows by transitivity from $a_{\mathsf{Q}}\vDash b$. Take any state $s$ in $\textsf{T}$-model $\mathbb{S}=(S,\sigma,V)$ with $s\Vdash_{\mathbb{S}}a$. Then $s\Vdash_{\mathbb{S}}a_{\mathsf{Q}}$ because $s$ is up-to-$p$ $L_{\textsf{P}}$-bisimular to itself for any $p\in\textsf{P}$, since $\Delta_S$ is an $L_{\textsf{P}}$-bisimulation.
\end{proof}

\end{document}